\documentclass[10pt, doublecolumn]{IEEEtran}
\usepackage{epsfig,latexsym}
\usepackage{float}
\usepackage{indentfirst}
\usepackage{amsmath}
\usepackage{amssymb}
\usepackage{times}
\usepackage{subfigure}
\usepackage{psfrag}
\usepackage{hyperref}
\usepackage{cite}
\usepackage{lastpage}
\usepackage{fancyhdr}
\usepackage{color}
 \usepackage{amsthm}
\usepackage{bigints}
\newtheorem{Lemma}{Lemma}
\newtheorem{lemma}[Lemma]{$\mathbf{Lemma}$}

\newcounter{problem}
\newcounter{save@equation}
\newcounter{save@problem}
\makeatletter
\newenvironment{problem}
{\setcounter{problem}{\value{save@problem}}%
  \setcounter{save@equation}{\value{equation}}%
  \let\c@equation\c@problem
  \subequations
}
{\endsubequations
  \setcounter{save@problem}{\value{equation}}%
  \setcounter{equation}{\value{save@equation}}%
}
 
\begin{document}

\title{ \vspace{-0.5em}{\huge Unveiling the Importance of SIC in NOMA Systems:
\\ Part II: New Results and Future Directions  }\\\vspace{-0.25em}
{\Large {\it (Invited Paper)}}}\vspace{-0.25em}

\author{ Zhiguo Ding, \IEEEmembership{Fellow, IEEE}, Robert Schober, \IEEEmembership{Fellow, IEEE}, and H. Vincent Poor, \IEEEmembership{Life Fellow, IEEE}    \thanks{ 
  
\vspace{-2em}

    Z. Ding and H. V. Poor are  with the Department of
Electrical Engineering, Princeton University, Princeton, NJ 08544,
USA. Z. Ding
 is also  with the School of
Electrical and Electronic Engineering, the University of Manchester, Manchester, UK (email: \href{mailto:zhiguo.ding@manchester.ac.uk}{zhiguo.ding@manchester.ac.uk}, \href{mailto:poor@princeton.edu}{poor@princeton.edu}).
R. Schober is with the Institute for Digital Communications,
Friedrich-Alexander-University Erlangen-Nurnberg (FAU), Germany (email: \href{mailto:robert.schober@fau.de}{robert.schober@fau.de}).

  }\vspace{-2em}}
 \maketitle
\begin{abstract}  
In most existing works on   non-orthogonal multiple access (NOMA), the decoding order of successive interference cancellation (SIC) is prefixed and based on either the users' channel conditions or their quality of service (QoS) requirements. A recent work on NOMA assisted semi-grant-free transmission showed  that  the use of a more sophisticated hybrid SIC scheme can yield significant performance improvements.  This letter   illustrates how the concept of hybrid SIC can be generalized and applied to different NOMA applications. We first use NOMA assisted mobile edge computing (MEC) as an example to illustrate the benefits of hybrid SIC, where new results for delay and energy minimization are presented. Then, future directions for generalizing  hybrid SIC  with adaptive decoding order selection as well as its promising applications  are discussed.  
\end{abstract} \vspace{-1.2em} 	
\section{Introduction}

Successive interference cancellation (SIC) is a key component of non-orthogonal multiple access (NOMA) systems, and is crucial  for	 the performance of NOMA transmission \cite{8972353,mojobabook,Zhiguo_CRconoma}. In the first part of this two-part invited paper, we have explained  that, in most existing works on NOMA, the design of the SIC decoding order is prefixed and  based on either the  users' channel state information (CSI) or their quality of service (QoS) requirements \cite{mojobabook,Zhiguo_CRconoma,8352630}. This is primarily due to the general perception  that  the use of more than one SIC decoding orders is trivial and  unnecessary.  In the first part of this paper,   the recent work in \cite{SGFx} on a hybrid implementation of    CSI- and QoS-based SIC  has also been reviewed, where we showed that   
adaptively switching between   CSI- and QoS-based SIC      can avoid an outage probability error floor,    which is inevitable with  either of the two individual   schemes.  

The aim of the second part of this paper is to show that the findings in \cite{SGFx} can be generalized  and can be applied to different NOMA communication scenarios. For  illustration, we use NOMA assisted mobile edge computing (MEC) as an example \cite{8267072,9022993,8794550,Zhiguo_MEC1}. Recall that the key idea of MEC is to ask users to offload their computationally intensive tasks to the base station, instead of computing these tasks locally. Compared to  orthogonal multiple access (OMA) based MEC, the use of NOMA-MEC  ensures that multiple users can offload their tasks simultaneously, which is beneficial for reducing  the delay and energy consumption of MEC offloading. New results for NOMA-MEC are presented in this letter by applying hybrid SIC. In particular, the problem of joint energy and delay minimization is considered, in order to demonstrate that the findings in \cite{SGFx} are useful not only for   performance analysis   but also for resource allocation. The optimal solution for joint energy and delay minimization is  obtained first, and  then compared to OMA-MEC and the existing NOMA-MEC solution \cite{Zhiguo_MEC1,8794550}. Furthermore, future directions for the design of sophisticated SIC schemes as well as     promising  applications in different NOMA communication scenarios are presented.  
 
\vspace{-0.5em}
 
 \section{System Model} 
 Consider a NOMA-MEC offloading scenario, where two users, denoted by  ${\rm U}_m$ and  ${\rm U}_n$, respectively, offload their computationally intensive and inseparable tasks to the base station.   ${\rm U}_i$'s channel gain and task deadline are denoted by $h_i$ and  $D_i$ seconds, $i\in\{m,n\}$, respectively.  
 It is assumed that  the users' tasks contain the same number of nats,   denoted by $N$. We note that unlike the first part of this paper which focuses on performance analysis, the second part of the paper concerns  resource allocation, where the use of  nats is more   convenient than the use of bits. We further assume that $D_m<D_n$, i.e., ${\rm U}_m$'s task is more delay sensitive    than  ${\rm U}_n$'s, which means that, in OMA-MEC,  ${\rm U}_m$ is served during the first $D_m$ seconds and then  ${\rm U}_n$ is served during the remaining $(D_n-D_m)$ seconds.  
 \vspace{-1em}
 \subsection{Basics of  NOMA-MEC  }
 Instead of allowing the first $D_m$ seconds to be solely occupied by  ${\rm U}_m$, NOMA-MEC encourages    that   ${\rm U}_n$    offloads a part of its task during the first $D_m$ seconds, and then the remainder  of its task during the following $T_n$ seconds, where $T_n\leq D_n-D_m$. Denote  ${\rm U}_n$'s transmit powers during the two time slots by $P_{n,1}$ and $P_{n,2}$, respectively. The advantage of NOMA-MEC over OMA-MEC can be illustrated  by considering the extreme case $(D_n-D_m)\rightarrow 0$.  In this case, ${\rm U}_n$'s transmit power in   OMA   has to be infinity in order to deliver $N$ nats in a short period, whereas this singular situation does not exist for NOMA-MEC since  ${\rm U}_n$ can also use  the first $D_m$ seconds for   offloading.
 \vspace{-1em}
 \subsection{Existing NOMA-MEC Strategies}
To ensure that the use of NOMA-MEC is transparent to  ${\rm U}_m$,     QoS-based SIC   has been used, i.e.,   ${\rm U}_n$'s signal is decoded before ${\rm U}_m$'s during the first $D_m$ seconds, where  ${\rm U}_n$'s data rate  during the first $D_m$ seconds needs to be constrained as $
 R_n = \ln \left(1+\frac{P_{n,1}|h_n|^2}{P_m|h_m|^2+1}\right)$ and $P_m$ denotes ${\rm U}_m$'s transmit power \cite{Zhiguo_MEC1,8794550}. Therefore, the problem of joint energy and delay minimization
can be formulated as follows:
 \begin{problem}\label{pb:0}
  \begin{alignat}{2}
\underset{T_n, P_{n,1}, P_{n,2}}{\rm{min}} &\quad   D_mP_{n,1}+T_nP_{n,2}\label{0obj:1} \\
\rm{s.t.} & \quad  D_mR_n +T_n\ln\left(1+|h_n|^2P_{n,2}\right)\geq N\label{0st:1}
\\
& \quad 0\leq T_n \leq D_n-D_m\label{0st:3}
\\
& \quad P_{n,i}\geq 0, \quad i\in\{1,2\}\label{0st:4},
  \end{alignat}
\end{problem} 
where
 constraints \eqref{0st:1} and  \eqref{0st:3}  ensure that  ${\rm U}_n$ can finish its offloading   within $D_n$ seconds. We note that we omit the costs for the computation at the base station as well as the costs for downloading the computation results from the base station, similar to \cite{8267072,9022993,8794550,Zhiguo_MEC1}.   Following the same steps as in \cite{Zhiguo_MEC1}, we can show that the optimal solution of $T_n$ is $T_n^*=D_n-D_m$, and the optimal power allocation solution is given by \vspace{-0.2em}
    \begin{eqnarray}\label{eqlema}
\left\{\hspace{-0.5em}\begin{array}{l}  
P_{n,1}^* = \frac{\left(1+P_m|h_m|^2\right)}{|h_n|^{2}}\left(e^{\frac{N-T_n\ln \left(1+P_m|h_m|^2\right)}{D_m+T_n}}-1\right)\\ 
P_{n,2}^* =  \frac{ \left(1+P_m|h_m|^2\right)e^{\frac{N-T_n\ln \left(1+P_m|h_m|^2\right)}{D_m+T_n}  }-1 }{|h_n|^{2}}
     \end{array}\right.\hspace{-0.5em},
\end{eqnarray}
if $P_m\leq |h_m|^{-2}\left(e^{\frac{N}{D_n-D_m}}-1\right)$, otherwise OMA is used. 

\vspace{-0.5em}

 \section{New NOMA-MEC with Hybrid SIC }
The aim of this section is to investigate whether there is any benefit in applying  hybrid SIC, i.e., selecting  the SIC orders in an adaptive manner, which means that   the problem of joint energy and delay minimization can be formulated as follows: 
  \begin{problem}\label{pb:x}\vspace{-1em}
  \begin{alignat}{2}
\underset{T_n, P_{n,1}, P_{n,2}}{\rm{min}} &\quad   D_mP_{n,1}+T_nP_{n,2}\label{xobj:1} \\
\rm{s.t.} & \quad  D_m R_{n,1}  +T_n\ln (1+P_{n,2}|h_n|^2)\geq N\label{xst:1}
\\
& \quad D_m \ln \left(1+\frac{P_m|h_m|^2}{P_{n,1}|h_n|^2+1}\right)\geq \mathbf{1}_nN\label{xst:2}
\\
&\quad   \eqref{0st:3}, \eqref{0st:4} .
  \end{alignat}
\end{problem} 
where $R_{n,1} =  \mathbf{1}_n\ln (1+P_{n,1}|h_n|^2) +(1-\mathbf{1}_n) \ln\left(1+\frac{P_{n,1}|h_n|^2}{P_m |h_m|^2+1}\right)$,  $ \mathbf{1}_n$ is the indicator function, i.e., $ \mathbf{1}_n=1$ if  ${\rm U}_m$'s signal is decoded first during the first $D_m$ seconds, otherwise $ \mathbf{1}_n=0$. We note that   \ref{pb:x} is degraded to   \ref{pb:0} if   $ \mathbf{1}_n=0$. Therefore, in the remainder of the letter, we focus on the case of $ \mathbf{1}_n=1$:
  \begin{problem}\label{pb:1}
  \begin{alignat}{2}
\underset{T_n, P_{n,1}, P_{n,2}}{\rm{min}} &\quad   D_mP_{n,1}+T_nP_{n,2}\label{obj:1} \\
\rm{s.t.} & \quad  D_m \ln (1+P_{n,1}|h_n|^2) \nonumber \\  &\quad+T_n\ln (1+P_{n,2}|h_n|^2)\geq N\label{st:1}
\\
& \quad D_m \ln \left(1+\frac{P_m|h_m|^2}{P_{n,1}|h_n|^2+1}\right)\geq N\label{st:2}
\\
& \quad \eqref{0st:3}, \eqref{0st:4} .
  \end{alignat}
\end{problem} 
The following lemma provides the optimal solution of   \ref{pb:1}.
\begin{lemma}\label{lemma1}
Assume $P_m>|h_m|^{-2}\left(e^{\frac{N}{D_m}}-1\right)$. 
For \ref{pb:1}, the optimal solution of $T_n$ is given by $T_n^*=D_n-D_m$. The optimal power allocation solution is given by \vspace{-0.5em}
   \begin{eqnarray}\label{hnoma1x}
\left\{\hspace{-0.5em}\begin{array}{l}  
P_{n,1}^* = |h_n|^{-2} \frac{P_m|h_m|^2}{e^{\frac{N}{D_m}}-1}-|h_n|^{-2}\\ 
P_{n,2}^* =  |h_n|^{-2}  e^{\frac{N}{D_n-D_m}-\frac{D_m}{D_n-D_m} \ln \left( \frac{P_m|h_m|^2}{e^{\frac{N}{D_m}}-1}\right)}-|h_n|^{-2}
     \end{array}\right.\hspace{-0.5em}\vspace{-0.5em},
\end{eqnarray}
if $   |h_m|^{-2}\left(e^{\frac{N}{D_m}}-1\right)< P_m \leq |h_m|^{-2}e^{  \frac{N}{D_n}}\left(e^{\frac{N}{D_m}}-1\right) $, otherwise
    \begin{eqnarray}\label{hnoma2x}
P_{n,1}^* =P_{n,2}^*= &|h_n|^{-2}\left(e^{\frac{N}{D_n}}-1\right).
\end{eqnarray}
 
\end{lemma}
\begin{proof}
See Appendix \ref{proof1}.
\end{proof}

 {\it Remark 1:}     Constraint \eqref{st:2} can be written as $P_{n,1}|h_n|^2\leq P_{m}|h_m|^2\left(e^{\frac{N}{D_m}}-1\right)^{-1}-1$. In order to ensure $P_{n,1}\neq 0$,   the feasibility of the constraint  needs the assumption      $P_m>|h_m|^{-2}\left(e^{\frac{N}{D_m}}-1\right)$ or equivalently $
D_m\ln (1+P_m|h_m|^2)> N$. Otherwise, OMA-MEC is used. In practice, this assumption can be justified     if  ${\rm U}_m$ is willing to increase its transmit power to help  ${\rm U}_n$. Also, if    ${\rm U}_m$ applies a coarse-level power control, $P_m$ has to be strictly larger than  $|h_m|^{-2}\left(e^{\frac{N}{D_m}}-1\right)$ anyways.   

{\it Remark 2:} The solutions of \ref{pb:0} and \ref{pb:1} share two common features. The first one is that they both outperform  OMA, as shown in   \cite{Zhiguo_MEC1} and in the proof for Lemma \ref{lemma1} in this letter.  The second one is that pure NOMA, i.e., $P_{n,2}=0$, is never preferred. In particular,  the solutions in   \eqref{eqlema},  \eqref{hnoma1x}, and \eqref{hnoma2x} correspond to the class of  hybrid NOMA schemes, i.e., ${\rm U}_n$ uses   NOMA     during the first $D_m$ seconds, and then   OMA   during  the  remaining $(D_n-D_m)$ seconds.   

The optimal solution of \ref{pb:x} can be straightforwardly obtained by numerically comparing the energy consumption required for the closed-form solutions in  \eqref{eqlema} and  \eqref{hnoma1x} (or \eqref{hnoma2x}), and selecting the most energy efficient solution. The  solutions in  \eqref{eqlema} and \eqref{hnoma2x} can be compared  analytically, as shown in the following lemma.  
\begin{lemma}\label{lemma2}
Assume $ e^{  \frac{N}{D_n}}\left(e^{\frac{N}{D_m}}-1\right)\leq   \left(e^{\frac{N}{D_n-D_m}}-1\right)$. For the case of $|h_m|^{-2}e^{  \frac{N}{D_n}}\left(e^{\frac{N}{D_m}}-1\right)\leq P_m \leq |h_m|^{-2}\left(e^{\frac{N}{D_n-D_m}}-1\right)$, the  new   solution shown in    \eqref{hnoma2x} is more energy efficient than the   existing   one shown in \eqref{eqlema}. 
\end{lemma} 
\begin{proof}
See Appendix \ref{proof2}.
\end{proof} 

{\it Numerical Studies:} 
In this section, the performance of different MEC strategies is studied by using computer simulations, where the users' average  channel gains are assumed to be identical  and normalized, a situation ideal  for the application of QoS-based SIC. We will show that it is still beneficial to use hybrid SIC in this situation.   In Fig. \ref{fig2a}, the energy consumption of MEC offloading is shown as a function of   $D_n$. As can be observed from the figure, the use of the new NOMA-MEC strategy can yield a significant reduction in energy consumption, compared to OMA-MEC and the existing  NOMA-MEC solution proposed in \cite{Zhiguo_MEC1}, particularly when $N$  is small.

   Fig. \ref{fig2a}  also shows that there are instances when the new NOMA-MEC scheme     achieves the same performance as the existing NOMA-MEC solution, which indicates that the solution of   \ref{pb:0} can outperform the one of   \ref{pb:1}. Therefore, in Fig. \ref{fig2b}, the solutions of \ref{pb:0} and \ref{pb:1} are compared in detail, where  $ P_m \geq |h_m|^{-2}\left(e^{\frac{N}{D_m}}-1\right)$ is considered.   When $P_m$ is small, the solution in \eqref{hnoma1x} is used, and Fig. \ref{fig2b} shows that it is possible for the solution of \ref{pb:0} to outperform the one of \ref{pb:1}.  By increasing $P_m$, the solution in \eqref{hnoma2x} becomes feasible, and Fig. \ref{fig2b} shows that the   solution  in \eqref{hnoma2x} is more energy efficient than the     one   in \eqref{eqlema}, which confirms Lemma \ref{lemma2}.    

\begin{figure}[t]  
\begin{center}\subfigure[ $P_m=1$ W ]{\label{fig2a}\includegraphics[width=0.38\textwidth]{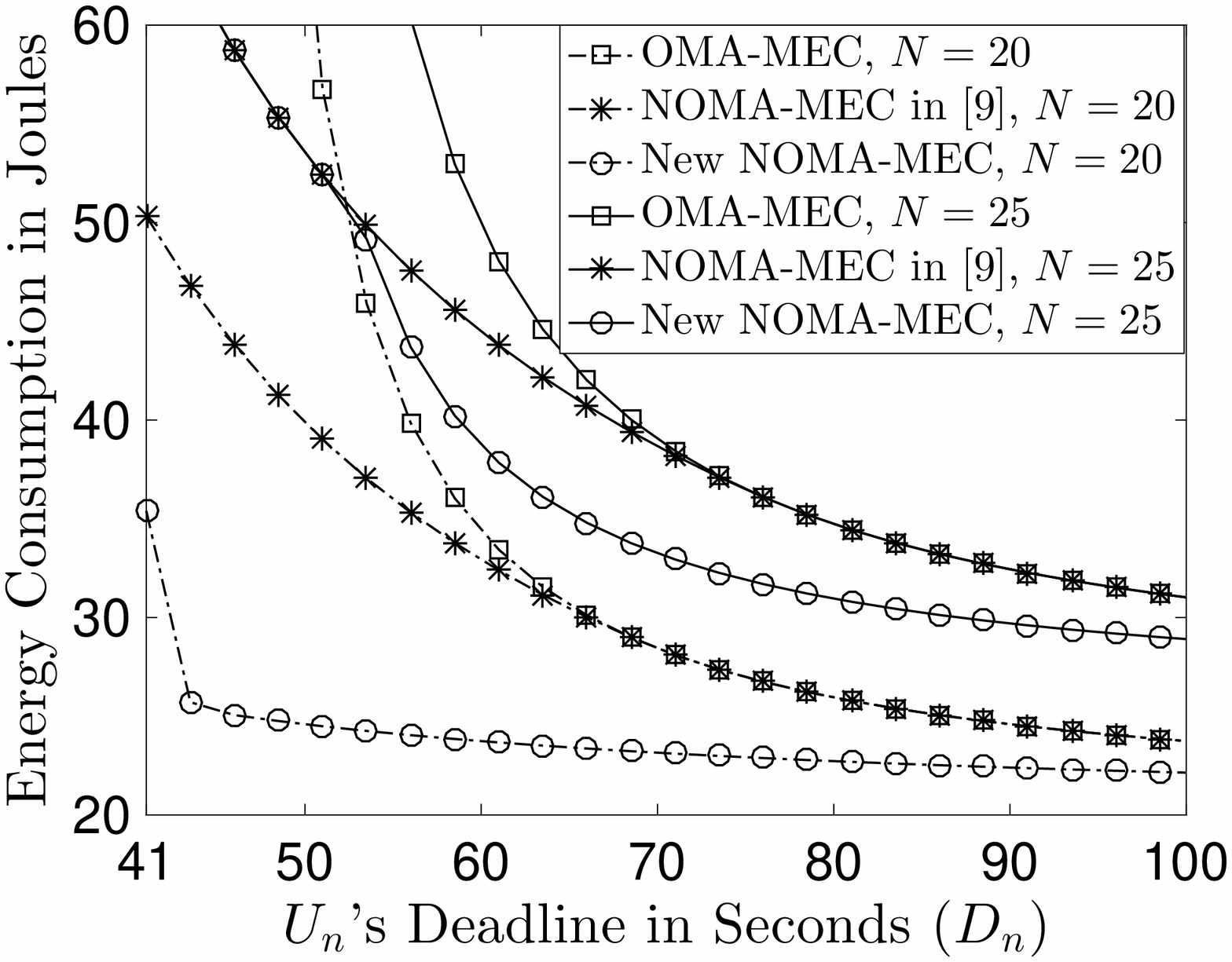}}
\subfigure[ $N=20$]{\label{fig2b}\includegraphics[width=0.38\textwidth]{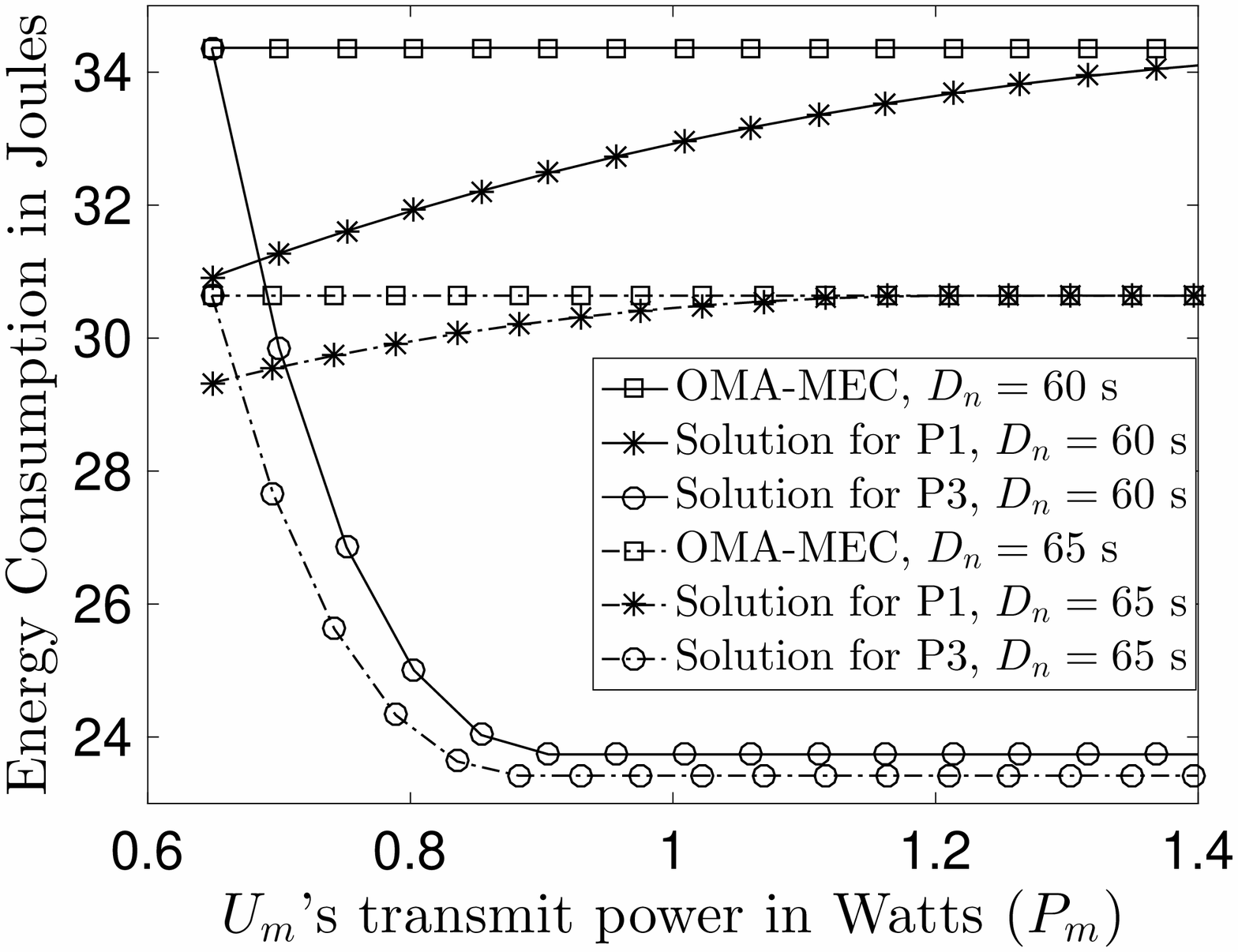}} \vspace{-1em}
\end{center} 
 \caption{The impact of the NOMA  strategies on the energy consumption required by MEC offloading.  $|h_m|^2= |h_n|^2=1$, and $D_m=40$~s.     }\label{fig 2}\vspace{-1.5em}
\end{figure}
\vspace{-1em}
\section{Conclusions and Future Directions}
In the second part of this invited paper, we have used NOMA-MEC as an example to illustrate how the new findings in \cite{SGFx} can be generalized. In particular,  a hybrid SIC based optimal solution for  joint energy and delay minimization was    obtained  and its superior performance compared   to   benchmark schemes was demonstrated.    Some promising directions for future research on hybrid SIC with adaptive decoding order selection are listed in the following.

\subsubsection{Fundamentals of hybrid SIC} For uplink NOMA, \cite{SGFx} showed the benefits of using hybrid SIC in two-user scenarios. When the number of users increases,   the number of possible SIC orders    increases significantly. Therefore, an important future direction is  to design practical hybrid SIC schemes  for striking a balanced tradeoff between system complexity and performance     \cite{7582543}. For downlink NOMA, it is still not known whether hybrid SIC  is beneficial, but the duality between uplink and downlink suggests that the design of hybrid SIC for downlink NOMA  is an important direction for future research.

\subsubsection{Green communications} The initial results shown in Fig.~\ref{fig2b} indicate that the use of hybrid SIC can significantly improve the energy efficiency of NOMA transmission. However, the energy reduction experienced by ${\rm U}_n$ is obtained at the price of increasing ${\rm U}_m$'s transmit power, which motivates a future study of user cooperation to improve the energy efficiency, which   opens up a new dimension for the design of future green communication systems.

\subsubsection{User clustering and resource allocation} For CSI-based SIC, it is preferable to group users with different channel conditions and encourage them to transmit/receive in the same subcarrier/time-slot. For  QoS-based SIC, it is preferable to group users with different QoS requirements.  These  clear  preferences provide simple guidances   for the design of user clustering and resource allocation. However,   hybrid SIC does not have these clear preferences, which makes a compact problem formulation difficult and   results in a  higher   complexity,   which is the price for the  significant performance improvements.  Therefore, designing low-complexity user clustering and resource allocation schemes  for   hybrid SIC is another important future research direction, where advanced tools, such as game theory and machine learning, can be   useful.

\subsubsection{Multiple-input multiple-output (MIMO) and intelligent reflecting surface (IRS) assisted NOMA} The use of hybrid SIC could be particularly useful in MIMO-NOMA systems. Recall that it is difficult  to order MIMO users due to the fact that the users' channels are in vector/matrix form. Therefore, most existing MIMO-NOMA schemes simply rely on the prefixed   SIC  decoding  order, whereas the use of hybrid SIC   increases the degrees of freedom available  for   system design.  Similarly, in the context  of IRS-NOMA, the use of hybrid SIC avoids relying on a single SIC decoding order, and hence introduces  more flexibility not only at the transceivers, but also at the IRS, which is helpful for  improving the system performance.   

\subsubsection{Emerging applications of NOMA} Many emerging  applications of NOMA will benefit from the use of hybrid SIC. For example, the delay and energy consumption of MEC offloading can be reduced, as shown by the initial results reported in this letter, but more rigorous studies from both the performance analysis and optimization perspectives  are needed. In addition to MEC, wireless caching is another functionality to be supported by fog networking, where hybrid SIC can also be useful. Particularly, in addition to the users' channel conditions and QoS requirements, the type of file content can also be taken into account  for the design of SIC. Similarly, in the context of NOMA assisted orthogonal time frequency space modulation (OTFS),   hybrid SIC can be further extended by taking   the users' heterogenous  mobility profiles into account for selecting the SIC decoding order. 

 \vspace{-1em}
\appendices
\section{Proof for Lemma \ref{lemma1}}\label{proof1}
 
\subsection{Obtaining Possible Solutions for Optimal Power Allocation}     

We   first find   closed-form solutions for power allocation by fixing $T_n$. 
By recasting constraint \eqref{st:1} as  $ -D_m \ln (1+P_{n,1}|h_n|^2) -T_n\ln (1+P_{n,2}|h_n|^2)\leq - N$,  it is straightforward to show that  \ref{pb:1} is convex, and the  optimal power allocation solution can be obtained by using the KKT conditions   listed in the following:
\begin{align}
\left\{ 
\begin{array}{l}
D_m - \lambda_3\frac{D_m|h_n|^2 }{1+P_{n,1}|h_n|^2} +\lambda_4|h_n|^2-\lambda_1=0\\
T_n-\lambda_3\frac{T_n|h_n|^2}{1+P_{n,2}|h_n|^2} -\lambda_2=0\\
\lambda_3(N-D_m \ln (1+P_{n,1}|h_n|^2) \\\qquad-T_n\ln (1+P_{n,2}|h_n|^2))=0\\
\lambda_4\left(P_{n,1}|h_n|^2 - \frac{P_m|h_m|^2}{e^{\frac{N}{D_m}}-1}+1   \right)=0\\
\lambda_iP_{n,i}=0, \quad i \in \{1,2\}
\\  \eqref{st:1}, \eqref{st:2}, \eqref{0st:3},\eqref{0st:4}
\end{array}
\right.,
\end{align}
where $\lambda_i$, $i\in\{1,\cdots, 4\}$, denote   Lagrange multipliers.  

Depending on the choices of the   Lagrange multipliers,  possible solutions are obtained as follows. 
\begin{itemize}
\item The choice of  $\lambda_1\neq0$ yields an OMA solution:
  \begin{eqnarray}\label{oma}
P_{n,1}^* = 0,\quad 
P_{n,2}^* =  |h_n|^{-2}\left( e^{\frac{N}{T_n}}-1\right).
\end{eqnarray}

\item The choice of $\lambda_1=0$, $\lambda_2=0$, and $\lambda_4\neq0$ yields a possible hybrid NOMA solution:
 \begin{eqnarray}\label{hnoma1}
\left\{\begin{array}{rl}  
P_{n,1}^* = &|h_n|^{-2}\left(\frac{P_m|h_m|^2}{e^{\frac{N}{D_m}}-1}-1 \right)\\ 
P_{n,2}^* = & |h_n|^{-2} \left(e^{\frac{N}{T_n}-\frac{D_m}{T_n} \ln \left( \frac{P_m|h_m|^2}{e^{\frac{N}{D_m}}-1}\right)}-1\right)
     \end{array}\right.,
\end{eqnarray}
if $ e^{  \frac{N}{D_m}}\left(e^{\frac{N}{D_m}}-1\right)\geq   P_m|h_m|^2\geq e^{\frac{N}{D_m}}-1$.

\item The choice of $\lambda_1=0$, $\lambda_2=0$, and $\lambda_4=0$ yields another possible hybrid NOMA solution:  
\begin{eqnarray}\label{hnoma2}
P_{n,1}^* =  
P_{n,2}^* =  |h_n|^{-2}\left(e^{\frac{N}{D_m+T_n}}-1\right),
\end{eqnarray}
if $   P_m|h_m|^2\geq  e^{\frac{N}{D_m+T_n}}\left(e^{\frac{N}{D_m}}-1\right) $.

\item The choice of $\lambda_1=0$ and $\lambda_2\neq0$ yields a pure NOMA solution:
 \begin{eqnarray}\label{noma}
P_{n,1}^* = |h_n|^{-2}\left(e^{\frac{N}{D_m}}-1\right), P_{n,2}^* =  0 ,
\end{eqnarray}
if $  P_m|h_m|^2\geq  e^{\frac{N}{D_m}}\left(e^{\frac{N}{D_m}}-1\right) 
$.
\end{itemize} 
\vspace{-1em}  
\subsection{Optimizing $T_n$}
Without loss of generality, take   the power allocation solution in \eqref{hnoma1} as an example. The corresponding  overall energy consumption is given by
\begin{align}  
E_{H1} =&D_m |h_n|^{-2}\left( \frac{P_m|h_m|^2}{e^{\frac{N}{D_m}}-1}-1\right)\\\nonumber&+T_n|h_n|^{-2}\left(e^{\frac{N}{T_n}-\frac{D_m}{T_n} \ln \left( \frac{P_m|h_m|^2}{e^{\frac{N}{D_m}}-1}\right)}-1\right).
 \end{align}
 By defining $N'=D_m \ln \left( \frac{P_m|h_m|^2}{e^{\frac{N}{D_m}}-1}\right)$, the overall energy consumption can be simplified  as follows:
 \begin{align}  \label{eh1xx}
E_{H1} =&\frac{D_m}{ |h_n|^{2}}\left( \frac{P_m|h_m|^2}{e^{\frac{N}{D_m}}-1}-1\right) +\frac{T_n}{|h_n|^{2}}\left(e^{\frac{N-N'}{T_n} }-1\right).
 \end{align}
 Define $f(x) \triangleq x\left(e^{\frac{a}{x}}-1\right)$ which is shown to be a monotonically decreasing function of $x$ for $x\geq 0$, where $a$ is a constant. The first order derivative of $f(x)$ is given by
\begin{align}
f'(x) 
=&\left(e^{\frac{a}{x}}-1\right)- e^{\frac{a}{x}}  \frac{a}{x}.
\end{align}
Further define $g(y) \triangleq \left(e^{y}-1\right)- ye^{y} $. One can find that $g(y)$ is a monotonically decreasing function of $y$ for $y\geq 0$, since 
\begin{align}
g'(y) =   &   e^{y} - e^{y}- ye^{y}=- ye^{y}\leq 0.
\end{align}
Therefore, $f'(x) $ is a monotonically increasing function of $x$, which means $f'(x) \leq  f'(\infty)=0 $, and hence $f(x)$ is indeed a monotonically decreasing function of $x$. Therefore, $T_n^*=D_n-D_m$ for the hybrid NOMA solution shown in \eqref{hnoma1}.  Similarly,   $T_n^*=D_n-D_m$  also holds for the other power allocation solutions. 
\vspace{-1em} 
\subsection{Comparison of the Solutions}
\subsubsection{ Comparing  the two hybrid NOMA solutions} For the case of $ e^{  \frac{N}{D_m}}\left(e^{\frac{N}{D_m}}-1\right)\geq   P_m|h_m|^2\geq e^{\frac{N}{D_n}}\left(e^{\frac{N}{D_m}}-1\right) $, the two hybrid NOMA solutions are feasible, and we will show that the solution in \eqref{hnoma2} outperforms the one in \eqref{hnoma1}.    

By using the fact that $T_n^*=D_n-D_m$,  the overall energy consumption for the solution in \eqref{hnoma2} is given by 
 \begin{align}\label{eh2}
 E_{H2} = D_n|h_n|^{-2}\left(e^{\frac{N}{D_n}}-1\right),
 \end{align}
 and the energy consumption of the   solution in \eqref{hnoma1} is given by \eqref{eh1xx}.   In order to show $E_{H1}\geq E_{H2}$,  it is sufficient to show that the following inequality holds  
\begin{align}\nonumber
&   D_n  e^{\frac{N}{D_n}}\left( \frac{P_m|h_m|^2}{e^{\frac{N}{D_m}}-1}\right)^{\frac{D_m}{D_n-D_m}} -D_m \left(\frac{P_m|h_m|^2}{e^{\frac{N}{D_m}}-1} \right)^{\frac{D_n}{D_n-D_m}}  \leq \\ 
& (D_n-D_m) e^{\frac{N}{D_n-D_m} }   .\label{ine1}
 \end{align}
To prove the   inequality in \eqref{ine1}, we define the following function
   \begin{align}  
\phi(x) =    D_n  e^{\frac{N}{D_n}}x^{\frac{D_m}{D_n-D_m}} -D_m x^{\frac{D_n}{D_n-D_m}}  ,
 \end{align} 
 where $e^{\frac{N}{D_n}}\leq x\leq e^{\frac{N}{D_m}}$. The first order derivative of $\phi(x)$ is given by
    \begin{align}  
\phi'(x) 
=  &    \frac{D_mD_n}{D_n-D_m}x^{\frac{D_m}{D_n-D_m}}  \left(e^{\frac{N}{D_n}}x^{-1} - 1 \right).
\end{align}
By using the fact that $x\geq e^{\frac{N}{D_n}}$, $\phi'(x) $ can be upper bounded as follows:
\begin{align} 
\phi'(x) \leq  &    \frac{D_mD_n}{D_n-D_m}x^{\frac{D_m}{D_n-D_m}}  \left(e^{\frac{N}{D_n}}\left(e^{\frac{N}{D_n}}\right)^{-1} - 1 \right)=0,
 \end{align} 
 which shows that $\phi(x)$ is a monotonically decreasing function of $x$ for $e^{\frac{N}{D_n}}\leq x\leq e^{\frac{N}{D_m}}$. Therefore, we have the following inequality
 \begin{align}  \nonumber
&   D_n  e^{\frac{N}{D_n}}\left( \frac{P_m|h_m|^2}{e^{\frac{N}{D_m}}-1}\right)^{\frac{D_m}{D_n-D_m}} -D_m \left(\frac{P_m|h_m|^2}{e^{\frac{N}{D_m}}-1} \right)^{\frac{D_n}{D_n-D_m}}   \\\nonumber
& \leq \phi\left(e^{\frac{N}{D_n}}\right)=D_n  e^{\frac{N}{D_n}}\left(e^{\frac{N}{D_n}}\right)^{\frac{D_m}{D_n-D_m}} -D_m \left(e^{\frac{N}{D_n}} \right)^{\frac{D_n}{D_n-D_m}}  
 \\ 
 & =D_n   e^{\frac{N}{D_n-D_m}} -D_m e^{\frac{N}{D_n-D_m}} .
 \end{align}
 Therefore, the inequality in \eqref{ine1} is proved, i.e., $E_{H1}\geq E_{H2}$ for 
$ e^{  \frac{N}{D_m}}\left(e^{\frac{N}{D_m}}-1\right)\geq   P_m|h_m|^2\geq e^{\frac{N}{D_n}}\left(e^{\frac{N}{D_m}}-1\right) $.

\subsubsection{Comparison of hybrid NOMA and pure NOMA} For the case of  $P_m|h_m|^2\geq
 e^{\frac{N}{D_m}}\left(e^{\frac{N}{D_m}}-1\right)$,   the pure NOMA solution in \eqref{noma} and   the hybrid NOMA solution in \eqref{hnoma2}  are feasible. The   energy consumption required by the pure NOMA solution is given by  
  \begin{align}
 E_{N} = &D_m|h_n|^{-2}\left(e^{\frac{N}{D_m}}-1\right) \\\nonumber \geq& D_n|h_n|^{-2}\left(e^{\frac{N}{D_n}}-1\right) = E_{H2},
 \end{align}
 where the inequality follows from the fact that $f(x)$ is a monotonically decreasing function of $x$ for $x\geq 0$.   

 \subsubsection{Comparison of OMA and hybrid NOMA}
 By following the same steps as in the previous subsection,   it is straightforward to show that the hybrid NOMA solution shown in \eqref{hnoma2} outperforms OMA. The comparison between OMA and the hybrid NOMA solution shown in \eqref{hnoma1} is challenging and will be focused on in the following.

Recall the energy consumption for OMA is  $
 E_{OMA}= (D_n-D_m) |h_n|^{-2}\left( e^{\frac{N}{D_n-D_m}}-1\right)$. In order to show $E_{H1}\leq E_{OMA}$, it is sufficient to prove the following inequality 
  \begin{align} \label{ine2}
 &D_m |h_n|^{-2}\left( \frac{P_m|h_m|^2}{e^{\frac{N}{D_m}}-1}-1\right)+(D_n-D_m)|h_n|^{-2}\\\nonumber&\times \left(e^{\frac{N}{D_n-D_m} }  \left( \frac{P_m|h_m|^2}{e^{\frac{N}{D_m}}-1}\right)^{-\frac{D_m}{D_n-D_m}}-1\right)\leq\\\nonumber
& (D_n-D_m) |h_n|^{-2}\left( e^{\frac{N}{D_n-D_m}}-1\right),
 \end{align}
 where $ |h_m|^{-2}\left(e^{\frac{N}{D_m}}-1 \right)\leq   P_m\leq |h_m|^{-2}e^{\frac{N}{D_n}}\left(e^{\frac{N}{D_m}}-1\right) $.
 
Eq. \eqref{ine2} is equivalent to the following inequality:  
  \begin{align} \label{ine2x}
 &D_m  \left( \frac{P_m|h_m|^2}{e^{\frac{N}{D_m}}-1}-1\right) +(D_n-D_m)  e^{\frac{N}{D_n-D_m} } \\\nonumber &\times \left[ \left( \frac{P_m|h_m|^2}{e^{\frac{N}{D_m}}-1}\right)^{-\frac{D_m}{D_n-D_m}} -1\right] \leq 0.
 \end{align}
 In order to prove \eqref{ine2}, we  define the following function
   \begin{align} 
 \varphi(x)=&D_m  \left( x-1\right) \\\nonumber &+(D_n-D_m)  e^{\frac{N}{D_n-D_m} } \left[ x^{-\frac{D_m}{D_n-D_m}} -1\right] . 
 \end{align}
The inequality in \eqref{ine2} can be proved if $\varphi(x)\leq 0$, for $  1\leq  x\leq  e^{\frac{N}{D_n}}  $, which is proved in the following.  
 The first order derivative of $\varphi(x) $ is given by
    \begin{align} 
 \varphi'(x) 
 =&D_m  -D_m  e^{\frac{N}{D_n-D_m} }  x^{-\frac{D_n}{D_n-D_m}}, 
 \end{align}
which shows that  $\varphi'(x)$ is a monotonically increasing function of $x$. By using the fact that $x\leq e^{\frac{N}{D_n}}$,   $\varphi'(x)$ can be lower bounded as follows:
     \begin{align} 
 \varphi'(x)\leq & \varphi'\left(e^{\frac{N}{D_n}}\right) 
 \\\nonumber
 =&D_m  -D_m  e^{\frac{N}{D_n-D_m} }  \left(e^{\frac{N}{D_n}}\right)^{-\frac{D_n}{D_n-D_m}}     =0. 
 \end{align}
 Therefore, $ \varphi(x)$ is a monotonically decreasing function of $x$. Since $x\geq 1$, we have
    \begin{align} 
 \varphi(x)\geq & \varphi(1)  =0,
 \end{align}
 which proves the inequality in \eqref{ine2}, i.e., $E_{OMA}>E_{H1}$. Therefore, hybrid NOMA outperforms pure NOMA and OMA, when all of them are feasible. When both the hybrid solutions are feasible, the solution in \eqref{hnoma2} outperforms the one in \eqref{hnoma1}. Thus, the proof is complete.  
 
 \vspace{-1em}
\section{Proof for Lemma \ref{lemma2}}\label{proof2}
For the   case of $|h_m|^{-2}e^{  \frac{N}{D_n}}\left(e^{\frac{N}{D_m}}-1\right)\leq P_m \leq |h_m|^{-2}\left(e^{\frac{N}{D_n-D_m}}-1\right)$,   both the two solutions in  \eqref{eqlema} and  \eqref{hnoma2x}  are feasible.  With some algebraic manipulations,  the overall energy consumption realized by the solution in  \eqref{eqlema} is given by
\begin{align}\nonumber
&E_0=D_m |h_n|^{-2} \left(
e^{\frac{N }{D_n}}
\left(1+P_m|h_m|^2\right)^{\frac{D_m}{D_n}}-P_m|h_m|^2-1\right)\\ \nonumber&+ (D_n-D_m) |h_n|^{-2}\left( 
e^{\frac{N }{D_n}  }
\left(1+P_m|h_m|^2\right)^{\frac{D_m}{D_n}  }
-1\right). 
\end{align}
The overall energy consumption with the solution in \eqref{hnoma2} is given by \eqref{eh2}. 
 In order to show that $E_0\geq E_{H2}$, it is sufficient to prove the following inequality:  
 \begin{align}\label{dmx}
-D_m    P_m|h_m|^2+ D_n
e^{\frac{N }{D_n}  }
\left(1+P_m|h_m|^2\right)^{\frac{D_m}{D_n}  }\geq 
 D_n   e^{\frac{N}{D_n}} .
\end{align}
 In order to prove \eqref{dmx}, we define the following function 
 \begin{align}
\psi(x) = -D_m    x+ D_n
e^{\frac{N }{D_n}  }
\left(1+x\right)^{\frac{D_m}{D_n}  } ,
\end{align}
where $e^{\frac{N}{D_n}}\left(e^{\frac{N}{D_m}}-1\right)\leq x \leq e^{\frac{N}{D_n-D_m}}-1$. The first order derivative of $\psi(x)$ is given by
 \begin{align}
\psi'(x) 
=&-D_m    + D_m
e^{\frac{N }{D_n}  } 
\left(1+x\right)^{\frac{D_m-D_n}{D_n}   } .
\end{align}
Because $D_m<D_n$, $\psi'(x) $ is a monotonically decreasing function of $x$. Given $x\leq e^{\frac{N}{D_n-D_m}}-1$, we have
 \begin{align}
\psi'(x)  \geq& \psi'\left(e^{\frac{N}{D_n-D_m}}-1\right)\\\nonumber
=&-D_m    + D_m
e^{\frac{N }{D_n}  } 
\left(e^{\frac{N}{D_n-D_m}}\right)^{\frac{D_m-D_n}{D_n}   }  =0,
\end{align}
which means that $\psi(x)$ is a monotonically increasing function of $x$ for $x\leq e^{\frac{N}{D_n-D_m}}-1$. Therefore, 
 \begin{align}
\psi(x) \geq \psi\left(e^{\frac{N}{D_n}}\left(e^{\frac{N}{D_m}}-1\right)\right) \geq \psi(0) =  D_n
e^{\frac{N }{D_n}  } .
\end{align}
Thus, \eqref{dmx} holds, i.e., $E_{0}\geq E_{H2}$. The proof is complete. 
     \bibliographystyle{IEEEtran}
\bibliography{IEEEfull,trasfer}

   \end{document}